\documentclass[aps,pra,onecolumn,english,superscriptaddress,showpacs,showkeys,nofootonbib,11pt,preprint]{revtex4-1}

\usepackage{amssymb}
\usepackage{amsmath}
\usepackage{amsthm}
\usepackage{amsfonts}
\usepackage{color}
\usepackage{amscd}
\usepackage{graphicx}
\usepackage{tikz}

\allowdisplaybreaks

\newtheorem{defin}{Definition}[section]
\newtheorem{prop}{Proposition}[section]

\newtheorem{lemma}{Lemma}[section]

\newtheorem{crit}{Criterion}[section]
\newtheorem{result}{Result}
\newtheorem{summary}{Summary}

\newcommand{\ket}[1]{|#1\rangle}
\newcommand{\bra}[1]{\langle #1 |}

\newcommand{\nn}{\nonumber}

\makeatother

\begin{document}

\title{A cone approach to the quantum separability problem}

\author{D. Salgado}
\email[]{dsalgado@nebrija.es}
\affiliation{D.G.\ Metodolog\'{\i}a, Calidad y Tecnolog\'{\i}as de la Informaci\'{o}n y las Comunicaciones, Insto.\ Nacional de Estad\'{\i}stica, 28071 Madrid (Spain)}
\affiliation{Dpto.\ Ingenier\'{\i}a Inform\'{a}tica, Universidad Antonio de Nebrija, 28040 Madrid (Spain)}

\author{J.L. S\'{a}nchez-G\'{o}mez}
\email[]{jl.sanchezgomez@uam.es}
\affiliation{Dpto.\ F\'{\i}sica Te\'{o}rica, Universidad Aut\'{o}noma de Madrid, 28049 Madrid (Spain)}

\author{M. Ferrero}
\email[]{maferrero@uniovi.es}
\affiliation{Dpto.\ F\'{\i}sica, Universidad de Oviedo, 33007 Oviedo (Spain)}

\date{\today}

\begin{abstract}
Exploiting the cone structure of the set of unnormalized mixed quantum states, we offer an approach to detect separability independently of the dimensions of the subsystems. We show that any mixed quantum state can be decomposed as $\rho=(1-\lambda)C_{\rho}+\lambda E_{\rho}$, where $C_{\rho}$ is a separable matrix whose rank equals that of $\rho$ and the rank of $E_{\rho}$ is strictly lower than that of $\rho$. With the simple choice $C_{\rho}=M_{1}\otimes M_{2}$ we have a necessary condition of separability in terms of $\lambda$, which is also sufficient if the rank of $E_{\rho}$ equals $1$. We give a first extension of this result to detect genuine entanglement in multipartite states and show a natural connection between the multipartite separability problem and the classification of pure states under stochastic local operations and classical communication (SLOCC). We argue that this approach is not exhausted with the first simple choices included herein.
\end{abstract}

\pacs{03.65.Ud, 03.67.Mn}

\keywords{Entanglement, Separability, Cone, Bipartite, Multipartite}

\maketitle

\section{Introduction}
\textquotedblleft Entanglement is [\dots] the characteristic trait of quantum mechanics\textquotedblright\ \cite{Sch35a}. This statement remains valid, even stronger, more than 70 years after its formulation. Nowadays entanglement is believed to lie behind most quantum phenomena and is still considered the source of conceptual challenges \cite{Gis09a}. In a few words, entanglement stands for the presence of quantum correlations, i.e. statistical correlations which cannot be reproduced using classical probability theory. Despite its importance, we are still far from a full comprehension of this fundamental feature of Nature, not to mention the fact that the problem of discerning whether a given composite quantum system portrays entanglement has only a partial solution. By and large, a complete understanding will follow only after apprehending the notion of quantum correlation both qualitatively and quantitatively for all sort of composite systems, either bipartite or multipartite and either with infinite or with finite dimensions.\\

Today a big gap exists between our understanding of bipartite and multipartite entanglement, probably as big as radical the difference is between having two subsystems and having more than two. This gap is also present between finite- and infinite-dimensional systems. These latter have been analyzed mainly in the Gaussian domain (see e.g. \cite{WanHirTomHay07a}). In the following we will concentrate upon finite-dimensional systems.  From a broad mathematical point of view two approaches can be followed to discern whether a given quantum state is entangled or not. On the one hand we can adopt an analytical point of view and try to find operational criteria to settle the question upon the quantum state itself. As prominent instances we can cite the Schmidt decomposition (only valid for pure states of any dimensions) \cite{NieChu00a}, the PPT criterion (valid as a necessary and sufficient condition for $2\otimes 2$ or $2\otimes 3$ and a necessary condition for the rest of bipartite systems) \cite{Wor76a,Per96a,HorHorHor96a}, the range criterion (a necessary condition for all bipartite systems) \cite{Hor97a} and the computable cross norm or realignment criterion (a necessary condition for all bipartite systems) \cite{Rud05a,CheWu03a}, to name a few (see \cite{LewBruCirKraKusSamSanTar00a,PleVir07a,HorHorHorHor09a,GuhTot09a} for an overview from different perspectives). On the other hand, we can adopt a numerical viewpoint, by which the separability problem for a given state is translated into an optimization problem and then an adequate algorithm is chosen to solve it \cite{DohParSpe02a,DohParSpe04a,Per04a}. Several of these approaches can be generalized to multipartite systems. A remarkable position in this view is played by entanglement witnesses \cite{Ter00a,LewKraCirHor00a}, which accept a double role in this division.\\

The numerical techniques can be understood to solve the question, but, in our view, this cannot be completely satisfactory inasmuch as a numerical solution to a second degree algebraic equation is not as satisfactory as its usual analytical solution. Furthermore we believe that such an analytical comprehension of the separability problem may help in related questions such as the distillability problem \cite{BruCirHorHulKraLewSan02a} and possibly open the way to understand entanglement in a much larger scale, in particular in quantum many-body systems (see e.g. \cite{AmiFazOstVed08a}). In this line we find it mandatory to pursue an analytical operational approach valid for any composite system in any dimensions.\\

Admittedly this is too ambitious a goal to cherish a realistic hope to find a global solution in a short term and we must restrict this generality. In this work we present the first steps aiming at finding an analytical operational solution to the bipartite entanglement problem in \emph{any} dimensions. Our proposal rests on the well-known structure of the set  of Hermitian positive semidefinite  matrices, namely, the convex cone structure $\textrm{PSD}_{N_{1}\otimes N_{2}}$ \cite{HilWat87a}, valid for any dimensions $N_{1}$ and $N_{2}$. Our bottom line tries both to be independent of the dimensions and of the bipartite character of the composite system, seeking a plausible generalization to multipartite systems. The price to be paid is that we have not covered all possible quantum states yet (see below). However we believe our results can be pushed forward to encompass the remaining cases and we hope to do this in a near future.\\

Our approach identifies the set of separable matrices $\textrm{SEP}_{N_{1}\otimes N_{2}}$ as a subcone of $\textrm{PSD}_{N_{1}\otimes N_{2}}$ and exploits the cone properties of this to write any unnormalized \footnote{Not necessarily of unit trace.} quantum state $Q\in\textrm{PSD}_{N_{1}\otimes N_{2}}$ as a convex combination $Q=(1-\lambda_{Q}) C_{Q}+\lambda_{Q} E_{Q}$ of a separable matrix $C_{Q}$ with the same rank as $Q$ and a matrix $E_{Q}$ in the boundary of $\textrm{PSD}_{N_{1}\otimes N_{2}}$. There are similar decompositions in the literature \cite{LewSan98a,KarLew01a} but with a different spirit. Our goal is to determine those values of $\lambda_{Q}\in[0,1]$ by which $Q$ is separable and those for which it is entangled. As a result, we will also have an explicit decomposition of $Q$ in terms of separable matrices. Regretfully for the time being we have only partially succeeded for those matrices whose range $\mathcal{S}=R(Q)$ can be written as a tensor product $\mathcal{S}=\mathcal{S}_{1}\otimes\mathcal{S}_{2}$. We also point out how this formalism admits a generalization to the multipartite case through the SLOCC entanglement classification of pure states \cite{DurVidCir00a,PleVir07a,LamLeoSalSol06c}.\\

The paper is organized as follows: in section \ref{ConStrPSD} we review those pertinent cone properties of the set of Hermitian positive semidefinite matrices. In section \ref{GenIde} we give the bottom line of our approach, which is applied in section \ref{FulRanQuaSta} to full-rank quantum states and in section \ref{QuaStaProFac} to quantum states lying in product faces. In section \ref{PlaGenMulSys} plausible arguments to generalize these techniques to the multipartite case are given as well as the natural connection with the SLOCC classification of pure states. We conclude with final remarks in section \ref{Con}.

\section{The cone structure of unnormalized quantum states}
\label{ConStrPSD}
A natural mathematical framework for finite-dimensional quantum states in the light of the separability problem is the cone of Hermitian positive semidefinite matrices. We review very briefly the notion of a cone and the properties of  $\textrm{PSD}_{N_{1}\otimes N_{2}}$, the cone of these matrices. This amounts to relaxing the separability problem to find a decomposition of an unnormalized quantum state $Q\in\textrm{PSD}_{N_{1}\otimes N_{2}}$ as

\begin{equation}
Q=\sum_{i=1}^{M}\lambda_{i}\ket{x_{i}}\bra{x_{i}}\otimes\ket{y_{i}}\bra{y_{i}},\qquad \lambda_i\geq 0\quad\forall i=1,\dots,M.
\end{equation}

Note that the non-negative parameters $\lambda_{i}$ are not restricted to satisfy $\sum_{i=1}^{M}\lambda_{i}=1$, as in the original formulation \cite{Wer89a}. We will very briefly include those notions pertinent to the forthcoming approach in order to render this work as self-contained as possible. The interested reader may consult the references \cite{Bar81a,HilWat87a} for more details.\\

\begin{defin}
Let $\mathcal{V}$ be a real topological vector space. A cone $K$ in $\mathcal{V}$ is a closed subset of $\mathcal{V}$ which satisfies (i) if $\alpha,\beta\geq 0$ and $\mathbf{x},\mathbf{y}\in K$, then $\alpha\mathbf{x}+\beta\mathbf{y}\in K$ and (ii) $K\cap (-K)=\emptyset$.
\end{defin}

The algebraic structure of a cone generalizes in some sense that of a vector space. In this line of thought, the analogous roles of the concepts of vector subspace and linear basis are played by those of \emph{faces} and \emph{extremals}, respectively.

\begin{defin}
A subcone $F\subset K$ is a face of the cone $K$ if $$\mathbf{x}\in K, \mathbf{y}-\mathbf{x}\in K,\textrm{ and }\mathbf{y}\in F\textrm{  imply  }\mathbf{x}\in F.$$
\end{defin}

\noindent Implicit in this definition is the fact that a face is a cone itself. Let $\dim F=\dim(F-F)$ and $\Phi(\mathcal{S})$ denote the least face of $K$ containing the vector subspace $S$.

\begin{defin}
A face $F$ is an extreme ray of the cone $K$ if $\dim F=1$. A vector $\mathbf{x}$ is an extremal if $\dim\Phi(\textrm{span}\{\mathbf{x}\})=1$.
\end{defin}

It will be often useful to identify a cone by its extremals, since any vector in the cone can be expressed as a so-called conic combination of extremals (except the extremals themselves), i.e. $\mathbf{y}=\sum_{i=1}^{n}\lambda_{i}\mathbf{x}_{i}$, where $\lambda_{i}\geq 0$ and $\mathbf{x}_{i}$ are extremals. Note how close is this expression to that of a linear combination in terms of elements of a linear basis. In this sense it is written $K=\textrm{cone}\{\mathbf{x}_{t}\}_{t\in T}$, where all vectors $\mathbf{x}_{t}$ are extremals and the index set $T$ may be non-numerable \footnote{When the set of extremals is finite, the cone is said to be polyhedral \cite{Bar81a}.}.\\

All these concepts apply readily to the set of Hermitian positive semidefinite matrices  acting upon the complex vector space $\mathbb{C}^{N_{1}}\otimes\mathbb{C}^{N_{2}}$. This set is a cone within the real vector space $\mathcal{H}_{N_{1}\otimes N_{2}}$ of Hermitian matrices denoted by $\textrm{PSD}_{N_{1}\otimes N_{2}}$. In particular, the extremals are matrices of the form $\lambda\ket{x}\bra{x}$, the extreme rays are subsets of the form $\{\lambda\ket{x}\bra{x}\}_{\lambda >0}$ and the faces are parametrized by the different vector subspaces belonging to $\mathbb{C}^{N_{1}}\otimes\mathbb{C}^{N_{2}}$ in such a way that given a vector subspace $\mathcal{S}$, then the face $F=F_{\mathcal{S}}$ is the set of all positive matrices (unnormalized quantum states) $Q$ such that $R(Q)\subseteq\mathcal{S}$, where $R(Q)$ denotes the range of $Q$ \cite{HilWat87a}. Note that we can write

\begin{equation}\label{PSDExt}
\textrm{PSD}_{N_{1}\otimes N_{2}}=\textrm{cone}\{\ket{x}\bra{x}\}_{\ket{x}\in\mathbb{C}^{1}\otimes\mathbb{C}^{2}}
\end{equation}

Indeed the spectral decomposition of any Hermitian positive semidefinite matrix is an example of this fact, since $Q=\sum_{i=1}^{N_{1}\cdot N_{2}}\lambda_{i}\ket{x_{i}}\bra{x_{i}}$, where $\lambda_{i}\geq 0$ are the eigenvalues of $Q$ and $\ket{x_{i}}$ the corresponding eigenvectors. Notice also that \eqref{PSDExt} expresses the fact that this kind of decompositions is valid without necessarily having $\ket{x_{i}}\perp\ket{x_{j}}$ for all $i\neq j$. Furthermore the faces, being cones themselves, can also be written as

\begin{equation}
F_{\mathcal{S}}=\textrm{cone}\{\ket{x}\bra{x}\}_{\ket{x}\in\mathcal{S}}.
\end{equation}

Important characteristics  of $\textrm{PSD}_{N_{1}\otimes N_{2}}$ to be exploited later are (see \cite{HilWat87a} and references therein):

\begin{enumerate}
\item The set of non-trivial faces (faces not being the whole cone $\textrm{PSD}_{N_{1}\otimes N_{2}}$ itself) constitutes the topological boundary of $\textrm{PSD}_{N_{1}\otimes N_{2}}$, which we will denote by $\partial\textrm{PSD}_{N_{1}\otimes N_{2}}$ \footnote{The topology in $\mathcal{H}_{N_{1}\otimes N_{2}}$ is the topology induced by the trace scalar product $(Q_{1},Q_{2})=\textrm{tr}(Q_{1}^{\dagger}Q_{2})=\textrm{tr}(Q_{1}Q_{2})$.}.
\item The preceding result is also valid for each face considered as a cone in the relative topology.
\item $\textrm{PSD}_{N_{1}\otimes N_{2}}$ and any of its faces are convex.
\item $\textrm{PSD}_{N_{1}\otimes N_{2}}$ is connected  and its faces are also connected in their relative topology.
\item Since every face $F=F_{\mathcal{S}}$ is characterized by a vector subspace $\mathcal{S}$ we can attach to each face an integer, namely $r=\dim\mathcal{S}$; this is indeed the rank of the greatest orthogonal projector $P_{\mathcal{S}}=P_{\mathcal{S}}^{\dagger}=P_{\mathcal{S}}^{2}$ belonging to $F_{\mathcal{S}}$, i.e. $r=r_{P_{\mathcal{S}}}=\dim\textrm{R}(P_{\mathcal{S}})$.
\end{enumerate}

As an example of the advantages of this approach, we offer a one-line proof that the cone of separable matrices has no facet \footnote{A facet of a cone of dimension $n$ is a face of dimension $n-1$.} \cite{GuhLut07a}: any face $F=F_{\mathcal{S}_{k}}$ ($\dim\mathcal{S}_{k}=k$) of $\textrm{PSD}_{N}$ has dimension $k^{2}$ \cite{HilWat87a}, thus any subcone of $\textrm{PSD}_{N_{1}\otimes N_{2}}$ has no facet, let alone the cone $\textrm{SEP}_{N_{1}\otimes N_{2}}$.

\section{The general idea}
\label{GenIde}
The proposal included herein rests upon the two following independent results. We will denote by $\textrm{int}\left(A\right)$ the topological interior of set $A$.

\begin{prop}\label{DecProp}
Let $Q\in\textrm{int}\left(F_{\mathcal{S}_{Q}}\right)$. Then there exists another matrix $C_{Q}\in\textrm{int}\left(F_{\mathcal{S}_{Q}}\right)$ and a matrix $E_{Q}\in\partial F_{\mathcal{S}_{Q}}$ such that

\begin{equation}\label{CEDec}
Q=\textrm{tr}(Q)\left[(1-\lambda_{Q})C_{Q}+\lambda_{Q}E_{Q}\right],
\end{equation}

\noindent where $\lambda_{Q}\in(0,1)$ and $\textrm{tr}(C_{Q})=\textrm{tr}(E_{Q})=1$.
\end{prop}

\begin{proof}
With no loss of generality we can suppose that $\textrm{tr}(Q)=1$ (otherwise repeat the forthcoming argument for $\frac{Q}{\textrm{tr}(Q)}$). Let $A_{1}$ denote the affine subspace of unit-trace Hermitian positive semidefinite matrices. Since $Q\in\left(\textrm{int}\left(F_{\mathcal{S}_{Q}}\right)\right)\cap A_{1}$ there always exists an open neighborhood of $Q$ contained in $\left(\textrm{int}\left(F_{\mathcal{S}_{Q}}\right)\right)\cap A_{1}$. Take $C_{Q}$ as any of the matrices belonging to this neighborhood. Since $\left(F_{\mathcal{S}_{Q}}\right)\cap A_{1}$ is connected it is possible to construct the closed segment $(1-\mu)C_{Q}+\mu Q$, with $\mu\in[0,1]$, lying inside $\left(\textrm{int}\left(F_{\mathcal{S}_{Q}}\right)\right)\cap A_{1}$ for all $\mu$ and to extend the segment ($\mu>1$) until $(1-\mu)C_{Q}+\mu Q\ \in\left(\partial F_{\mathcal{S}_{Q}}\right)\cap A_{1}$. Let $\mu^{*}$ be the minimal $\mu>1$ such that $(1-\mu)C_{Q}+\mu Q\ \in\left(\partial F_{\mathcal{S}_{Q}}\right)\cap A_{1}$. Then $E_{Q}=(1-\mu^{*})C_{Q}+\mu^{*} Q\ \in\left(\partial F_{\mathcal{S}_{Q}}\right)\cap A_{1}$ and, given the convexity of $F_{\mathcal{S}_{Q}}\cap A_{1}$ there always exist a real $\lambda\in(0,1)$ such $Q=(1-\lambda)C_{Q}+\lambda E_{Q}$.
\end{proof}

Note that decomposition \eqref{CEDec} is similar to others pursued in different approaches \cite{LewSan98a,KarLew01a}. The resemblance arises, in our view, from the convex cone structure of the set of quantum states; the difference stems from the diverse alternative we have when choosing either $C_{Q}$ or $E_{Q}$. This decomposition is rather easy to compute from a theoretical point of view: find the minimal $\mu>1$ such that $r_{(1-\mu)C_{Q}+\mu Q}< r_{Q}$.\\

The second result we need is the following

\begin{prop}\label{SepProp}
The cone of separable matrices $\textrm{SEP}_{N_{1}\otimes N_{2}}$ is a proper subcone of $\textrm{PSD}_{N_{1}\otimes N_{2}}$.
\end{prop}

\begin{proof}
This elementarily follows from the fact that the set of extremals $\{\ket{x}\bra{x}\otimes\ket{y}\bra{y}\}_{\ket{x}\in\mathbb{C}^{N_{1}},\ket{y}\in\mathbb{C}^{N_{2}}}$ of $\textrm{SEP}_{N_{1}\otimes N_{2}}$ is strictly contained in the set of extremals $\{\ket{z}\}_{\ket{z}\in\mathbb{C}^{N_{1}}\otimes\mathbb{C}^{N_{2}}}$ of $\textrm{SEP}_{N_{1}\otimes N_{2}}$.
\end{proof}

The general idea of this approach lies on the choice of $C_{Q}$ in the decomposition \eqref{CEDec} in the light of proposition \ref{SepProp}:

\begin{crit}
Given a quantum state $\rho$ choose $C_{\rho}\in\textrm{SEP}_{N_{1}\otimes N_{2}}$ and let $\lambda_{\rho}^{*}$ be the greatest $\lambda_{\rho}$ such that $(1-\lambda_{\rho}^{*})C_{\rho}+\lambda_{\rho}^{*}E_{\rho}\in\textrm{SEP}_{N_{1}\otimes N_{2}}$. Then $\rho$ will be separable if and only if $\lambda_{\rho}\leq\lambda_{\rho}^{*}$.
\end{crit}

\begin{proof}
The result follows readily from the convexity of both $\textrm{SEP}_{N_{1}\otimes N_{2}}$ and $\textrm{PSD}_{N_{1}\otimes N_{2}}$ and from the fact that $\textrm{SEP}_{N_{1}\otimes N_{2}}\subset\textrm{PSD}_{N_{1}\otimes N_{2}}$. Since both $C_{\rho}$ and $(1-\lambda_{\rho}^{*})C_{\rho}+\lambda_{\rho}^{*}E_{\rho}$ are separable, then $(1-\lambda_{\rho})C_{\rho}+\lambda_{\rho}E_{\rho}$ is separable whenever $\lambda_{\rho}\in[0,\lambda_{\rho}^{*}]$, and for $\lambda_{\rho}>\lambda_{\rho}^{*}$, by construction, $(1-\lambda_{\rho})C_{\rho}+\lambda_{\rho}E_{\rho}$ is not separable.
\end{proof}

Note that given a quantum state $\rho$ this separability criterion is reduced to the face $F_{\mathcal{S}_{\rho}}$, thus not involving quantum states belonging to the rest of the cone. Although not indicated in the notation, the $\lambda^{*}_{\rho}$ values depend also on $C_{\rho}$ and $E_{\rho}$. Notice also that once the values $\lambda^{*}_{\rho}$ are determined, the separability check consists of a simple comparison between the value $\lambda_{\rho}$ in the decomposition \eqref{CEDec} and the corresponding $\lambda^{*}_{\rho}$ value. In the subsequent sections we will concentrate upon the computation of these values.

\section{Full-rank quantum states}
\label{FulRanQuaSta}

Firstly we will concentrate upon finding the $\lambda^{*}_{\rho}$ values for those quantum states $\rho$ having full rank. Following proposition \ref{DecProp} we  have always the freedom to choose either $C_{\rho}$ or $E_{\rho}$. Our choice is $C_{\rho}=M_1\otimes M_2$, where $M_i\in\textrm{int\ PSD}_{N_{i}}$, that is, $r_{M_{i}}=N_{i}$ for $i=1,2$. The resulting $E_{\rho}$ in decomposition \eqref{CEDec} is not determined a priori, so  it could be any matrix with rank strictly lower that $N_{1}\cdot N_{2}$. We will distinguish between the cases in which $r_{E_{\rho}}=1$ and those in which $r_{E_{\rho}}>1$.\\

We need previously some notation and some preliminary results. Given a pure quantum state $\ket{z}\in\mathbb{C}^{N_{1}}\otimes\mathbb{C}^{N_{2}}$, we define the $(N_{1}\times N_{2})$-dimensional matrix of coefficients in the computational basis $\ket{z}=\sum_{i_{1}i_{2}}c_{i_{1}i_{2}}\ket{i_{1} i_{2}}$ by $C(z)=[c_{i_{1}i_{2}}]$, where $i_{1}$ is a row column and $i_{2}$ a column index. Then we can apply any generalized singular value decomposition (GSVD) to $C(z)$ so that $C(z)=F(z)D(z)G^{\dagger}(z)$, where $F(z)$ is an $N_{1}$-dimensional square nonsingular matrix, $D(z)$ is a $(N_{1}\times N_{2})-$dimensional diagonal matrix with nonnull entries $\sigma_{0}(z)\geq\dots\sigma_ {r-1}(z)>0$ and $G(z)$ is an $N_{2}$-dimensional square nonsingular matrix (see appendix \ref{GSVD} for an overview of the GSVD). For ease of notation we do not specify the particular GSVD yielding the generalized singular values $\sigma_{i}(z)$, this being clear from the context. When the proper SVD is indeed applied, the singular values are in fact the Schmidt coefficients of $\ket{z}$, thus we suggest to name them generalized Schmidt coefficients. Notice that $\sum_{i=0}^{r-1}\sigma_{i}^{2}(z)=1$, since the pure quantum state $\ket{z}$ is normalized. The number $r$ of non-null generalized Schmidt coefficients is the rank of $C(z)$. Note that the Schmidt criterion to decide whether a given pure quantum state is entangled or not is readily contained in this formulation, since $\ket{z}$ is separable if and only if $r_{C(z)}=1$, which is equivalent to $\sigma_{1}(z)=0$ \cite{LamLeoSalSol06c}.\\

In matrix notation this all is elementarily formulated as follows:

\begin{prop}
Let $\ket{z}\bra{z}\in\textrm{PSD}_{N_{1}\otimes N_{2}}$. Then
\begin{enumerate}
\item $\ket{z}\bra{z}=\left(F(z)\otimes G^{*}(z)\right)\ket{\sigma}\bra{\sigma}\left(F^{\dagger}(z)\otimes G^{t}(z)\right)$, where $F(z)$ ($G(z)$) is the left (right) matrix of the GSVD of $C(z)$ and $\ket{\sigma}=\sum_{i=0}^{r-1}\sigma_{i}(z)\ket{ii}$.
\item $\ket{z}\bra{z}\in\textrm{SEP}_{N_{1}\otimes N_{2}}$ if, and only if, $\sigma_{1}(z)=0$.
\end{enumerate}
\end{prop}

\subsection{Full-rank states with $r_{E_{\rho}}=1$}

We will find those nonnegative $\lambda$ such that $\rho=(1-\lambda)M_1\otimes M_2+\lambda\ket{z}\bra{z}\in\textrm{SEP}_{N_{1}\otimes N_{2}}$, i.e. we will compute $\lambda^{*}_{\rho}$. Since $r_{\rho}=N_{1}\cdot N_{2}$, $M_{1}$ and $M_{2}$ are also full-rank matrices, hence there always exist nonsingular matrices $F_{1}$ and $F_{2}$ such that $M_{i}=F_{i}F_{i}^{\dagger}$ \cite{HorJoh91a}. We have the following partial result:

\begin{prop}
Let $\rho=(1-\lambda)M_1\otimes M_2+\lambda\ket{z}\bra{z}$ with $M_{i}=F_{i}F_{i}^{\dagger}$. Then $\rho\in\textrm{SEP}_{N_{1}\otimes N_{2}}$ if, and only if, $\bar{\rho}\equiv(1-\lambda)\frac{\mathbb{I}_{N_{1}N_{2}}}{N_{1}N_{2}}+\lambda\ket{\sigma}\bra{\sigma}\in\textrm{SEP}_{N_{1}\otimes N_{2}}$, where $\ket{\sigma}=\sum_{i=0}^{r-1}\sigma_{i}(z)\ket{ii}$ is built with the generalized Schmidt coefficients of $C(z)$ with respect to $(M^{-1}_{1},M^{-1}_{2})$.
\end{prop}

The first key result is the following necessary and sufficient separability criterion for $\bar{\rho}$:

\begin{prop}\label{ExtPro}
Let $\bar{\rho}\equiv(1-\lambda)\frac{\mathbb{I}_{N_{1}N_{2}}}{N_{1}N_{2}}+\lambda\ket{\sigma}\bra{\sigma}$. Then $\bar{\rho}\in\textrm{SEP}_{N_{1}\otimes N_{2}}$ if, and only if, $\lambda\leq\lambda^{*}(z)\equiv\frac{1}{1+N_{1}N_{2}\sigma_{0}(z)\sigma_{1}(z)}$.
\end{prop}

\begin{proof}
The detailed proof is rather technical and cumbersome, thus we have moved some technical lemmas to appendix \ref{Lem} for clarity's sake. We make use of the partial transpose, whose positivity is a necessary condition of separability \cite{Wor76a,Per96a}. By lemma \ref{PPT} $\bar{\rho}^{T_{1}}$ is positive if, and only if, $\lambda\leq \lambda^{*}(z)\equiv \frac{1}{1+N_{1}N_{2}\sigma_{0}(z)\sigma_{1}(z)}$. Thus we must find a criterion to decide for which $\lambda\in[0,\lambda^{*}(z)]$ we can conclude that $\bar{\rho}$ is separable. Since it is clearly separable for $\lambda=0$, if we prove the separability for $\lambda=\lambda^{*}(z)$, then from the convexity of $\textrm{PSD}_{N_{1}\otimes N_{2}}$ and hence of $\textrm{SEP}_{N_{1}\otimes N_{2}}$ our result follows. The separability of $(1-\lambda^{*}(z))\frac{\mathbb{I}_{N_{1}N_{2}}}{N_{1}N_{2}}+\lambda^{*}(z)\ket{\sigma}\bra{\sigma}$ is proved in lemma \ref{SepWerner}.
\end{proof}

This result is interesting in itself even without relation to the whole approach. Firstly, the state $\bar{\rho}$ is a more realistic mathematical description of the experimentally prepared pure quantum state $\ket{\sigma}$ (hence of any pure state because there always exist unitary matrices $U$ and $V$ such that $\rho=(1-\lambda)\frac{\mathbb{I}_{N_{1}N_{2}}}{N_{1}N_{2}}+\lambda|\Psi\rangle\langle\Psi|=U\otimes V^{*}\left((1-\lambda)\frac{\mathbb{I}_{N_{1}N_{2}}}{N_{1}N_{2}}+\lambda\ket{\sigma}\bra{\sigma}\right)U^{\dagger}\otimes V^{t}$ for every $\ket{\Psi}$), since it is ineludible to introduce some noise in the state preparation procedure. Thus states like $\bar{\rho}$ are relevant per se; they are indeed Werner states in arbitrary dimensions \cite{PitRub00a}. Secondly, this result will allow us to generalize the sufficient and necessary condition of separability already proved for some generalized Werner states \cite{PitRub00a,DenChe08a} in the multipartite realm (see below).\\

Merging both results we have a first sufficient and necessary criterion of separability for bipartite states of arbitrary dimensions of the form $\rho=(1-\lambda)M_1\otimes M_2+\lambda\ket{z}\bra{z}$:

\begin{result}\label{ResSufNec}
Let $\rho=(1-\lambda)M_1\otimes M_2+\lambda\ket{z}\bra{z}$, with $M_{i}=F_{i}F_{i}^{\dagger}$, $r_{M_{i}}=N_{i}$, $i=1,2$, and $\lambda\in[0,1]$. Then $\rho$ is separable if, and only if, $\lambda\leq\lambda^{*}(z)$, where $\lambda^{*}(z)=\frac{1}{1+N_{1}N_{2}\sigma_{0}(z)\sigma_{1}(z)}$ and $\sigma_{i}(z)$ are the generalized Schmidt coefficients of $\ket{z}$ with respect to $(M_{1}^{-1},M_{2}^{-1})$.
\end{result}

\subsection{Full-rank states with $r_{E_{\rho}}>1$}

Following along the same line as above, we focus on those cases in which the decomposition \eqref{CEDec} yields a matrix $E_{\rho}$ such that $r_{E_{\rho}}>1$. We arrive at the first incomplete result of this approach.

\begin{result}\label{IncomRes}
Let $\rho=(1-\lambda)M_1\otimes M_2+\lambda E_{\rho}\in\textrm{int}\ \textrm{PSD}_{N_{1}\otimes N_{2}}$ where $M_{i}=F_{i}F_{i}^{\dagger}$, $M_{1}\otimes M_{2}\in\textrm{int}\ \textrm{PSD}_{N_{1}\otimes N_{2}}$, $\lambda\in[0,1]$ and $E_{\rho}=\sum_{k=1}^{K}e_{k}\ket{e_{k}}\bra{e_{k}}$. Then $\rho\in\textrm{SEP}_{N_{1}\otimes N_{2}}$ if $\lambda\leq\bar{\lambda}(E_{\rho})\equiv\frac{1}{\sum_{k=1}^{K}\frac{e_{k}}{\lambda^{*}(e_{k})}}$, where $\lambda^{*}(e_{k})=\frac{1}{1+N_{1}N_{2}\sigma_{0}(e_{k})\sigma_{1}(e_{k})}$.
\end{result}

\begin{proof}
Let $\mu_{k}=\frac{\frac{e_{k}}{\lambda_{k}^{*}}}{\sum_{k=1}^{K}\frac{e_{k}}{\lambda_{k}^{*}}}$ and $\bar{\rho}=\rho|_{\lambda=\bar{\lambda}(E_{\rho})}$. Given that $\sum_{k=1}^{K}e_{k}=1$ (since $\textrm{tr}(E_{\rho})=1$), it is immediately to prove that $$\bar{\rho}=\sum_{k=1}^{K}\mu_{k}\left[(1-\lambda^{*}(e_{k}))M_{1}\otimes M_{2}+\lambda^{*}(e_{k})\ket{e_{k}}\bra{e_{k}}\right].$$

Now this a convex combination of separable matrices, hence it is separable. Thus if $\rho$ is separable when $\lambda=0$ and when $\lambda=\bar{\lambda}(E_{\rho})$, then from the convexity of $\textrm{SEP}_{N_{1}\otimes N_{2}}$ we conclude that $\rho$ is separable whenever $\lambda\leq\bar{\lambda}(E_{\rho})$.
\end{proof}

Regretfully this is an incomplete result because the condition is not sufficient, as the following counterexample shows. Let $\rho=(1-\lambda)\frac{\mathbb{I}_{2}\otimes\mathbb{I}_{ 2}}{4}+\lambda\left(\frac{1}{2}\ket{01}\bra{01}+\frac{1}{2}\ket{\Psi^+}\bra{\Psi^+}\right)$, where $\ket{\Psi^+}=\frac{1}{\sqrt{2}}\left(\ket{00}+\ket{11}\right)$. Then $\bar{\lambda}(E_{\rho})=\frac{1}{2}$, while the PPT condition states that it is separable for all $\lambda\leq\frac{1}{\sqrt{2}}$.\\

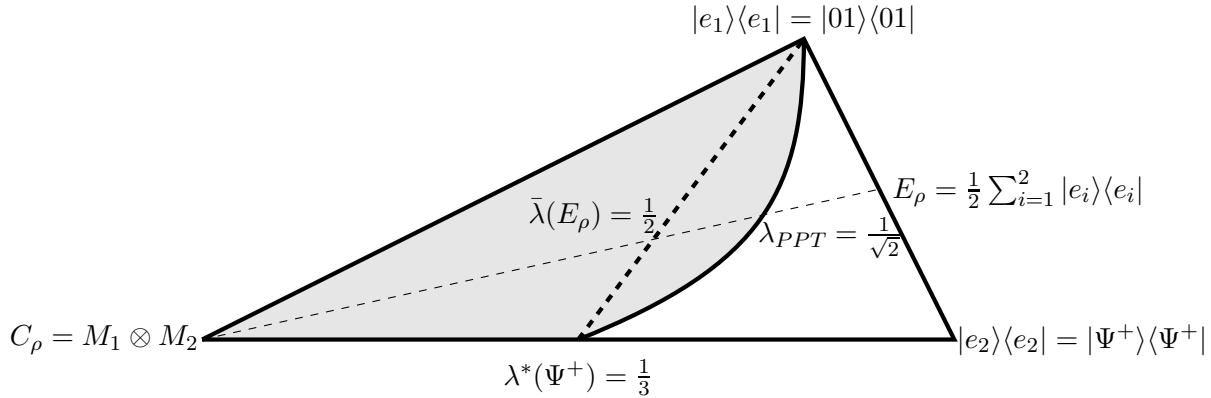
\begin{figure}
\begin{center}
\begin{tikzpicture}
\fill[gray!20!white] (0,0) -- (5.0,0.0) .. controls (7.5,1.0) and (8.0,2.0) .. (8.0,4.0) -- (0,0);
\draw[ultra thick] (0,0)--(10.0,0)--(8.0,4.0)--(0,0);
\draw[ultra thick] (5.0,0.0) .. controls (7.5,1.0) and (8.0,2.0) .. (8.0,4.0);
\draw[dashed,ultra thick] (5.0,0.0)--(8.0,4.0);
\draw[dashed] (0,0)--(9.0,2.0);
\draw (-1.30,0) node {$C_{\rho}=M_{1}\otimes M_{2}$};
\draw (10.85,2.0) node {$E_{\rho}=\frac{1}{2}\sum_{i=1}^{2}\ket{e_{i}}\bra{e_{i}}$};
\draw (8.0,4.25) node {$\ket{e_{1}}\bra{e_{1}}=\ket{01}\bra{01}$};
\draw (11.70,0) node {$\ket{e_{2}}\bra{e_{2}}=\ket{\Psi^{+}}\bra{\Psi^{+}}$};
\draw (5.20,1.65) node {$\bar{\lambda}(E_{\rho})=\frac{1}{2}$};
\draw (8.35,1.35) node {$\lambda_{PPT}=\frac{1}{\sqrt{2}}$};
\draw (5.0,-0.5) node {$\lambda^{*}(\Psi^{+})=\frac{1}{3}$};
\end{tikzpicture}
\end{center}
\caption{\label{ConSec} Bidimensional section of both the cone $\textrm{SEP}_{2\otimes 2}$ (shaded) and the cone $\textrm{PSD}_{2\otimes 2}$ illustrating result \eqref{IncomRes} through a counterexample for its sufficiency (not scaled).}
\end{figure}

This is better illustrated in a geometric fashion (see figure \ref{ConSec}). Result \ref{IncomRes} does not achieve to detect those separable states $\rho=(1-\lambda)C_{\rho}+\lambda E_{\rho}$ with $\lambda\in(\frac{1}{2},\frac{1}{\sqrt{2}}]$ but only those within the \textquotedblleft polygonal\textquotedblright approximation to the cone section generated by $C_{\rho}$, $\ket{e_{1}}\bra{e_{1}}$ and $\ket{e_{2}}\bra{e_{2}}$. Thus far we do not have a clear hint of how to overcome this difficulty arising from the highly non-polyhedral character of $\textrm{SEP}_{N_{1}\otimes N_{2}}$.\\

However we leave as an open question the following: given a full-rank state $\rho$, does there always exist a full-rank product matrix $C_{\rho}=M_{1}\otimes M_{2}$ such that the rank of $E_{\rho}$ is $1$, i.e. $r_{E_{\rho}}=1$? In the light of result \ref{ResSufNec} an affirmative answer would solve completely the separability problem for full-rank states. On the contrary, a negative answer would entail important consequences to result \ref{IncomRes}: the \textquotedblleft polygonal\textquotedblright approximation must be overcome.

\section{Quantum states in product faces}
\label{QuaStaProFac}
In this section we extend the preceding results to those quantum states not having full rank. Notice that this approach is independent of the dimensions of the factors in the tensor product structure and in particular that the decomposition \eqref{CEDec} is valid for any face $F_{\mathcal{S}}$ of the cone $\textrm{PSD}_{N_{1}\otimes N_{2}}$. We need the orthogonal complementation of a matrix:

\begin{defin}
Let $F=\sum_{i=0}^{r-1}\ket{c_{i}}\bra{i}$ be an $N$-dimensional square matrix with rank $r$. We define the orthogonal complementation of $F$, and denote it by $\bar{F}$, as the matrix $\bar{F}=\sum_{i=0}^{r-1}\ket{c_{i}}\bra{i}+\sum_{i=r}^{N-r}\ket{c^{\perp}_{i}}\bra{i}$, where the vectors $\{\ket{c^{\perp}_{i}}\}_{i=r,\dots,N-r}$ spans the orthogonal subspace to $R(F)$.
\end{defin}

We agree to define $\bar{F}=F$ whenever $r=N$. Notice that there always exists $\bar{F}^{-1}$. This definition allows us to write the following identity:

\begin{prop}
Let $M=FF^{\dagger}$, with $r_{M}=r$, be an $N$-dimensional matrix such that $M=\sum_{i=0}^{r-1}\ket{m_{k}}\bra{m_{k}}$, where the kets $\ket{m_{k}}$ are unnormalized. Then $$M=FF^{\dagger}=\bar{F}\left(\mathbb{I}_{r}\oplus\mathbb{O}_{N-r}\right)\bar{F}^{\dagger},$$
\noindent where $\mathbb{O}_{m}$ denotes the $m-$dimensional square null matrix and $\oplus$ denotes the direct sum of matrices.
\end{prop}

\begin{proof}
Since we can write $F=\sum_{i=0}^{r-1}\ket{m_{k}}\bra{k}$, we have $\bar{F}=\sum_{i=0}^{r-1}\ket{m_{k}}\bra{k}+\sum_{i=r}^{N-1}\ket{m_{k}^{\perp}}\bra{k}$, and the result follows readily.
\end{proof}

We can formulate the following

\begin{result}\label{ProFaces}
Let $\rho=(1-\lambda)M_1\otimes M_2+\lambda E_{\rho}$ with $M_{i}=F_{i}F_{i}^{\dagger}$, $r_{M_{i}}=r_{i}$, $i=1,2$, $r_{\rho}=r_{1}\cdot r_{2}$, $\lambda\in[0,1]$ and $E_{\rho}=\sum_{k=1}^{K}e_{k}\ket{e_{k}}\bra{e_{k}}$. Let $\bar{F}_{i}$ be the orthogonal complementation of matrix $F_{i}$, $i=1,2$. Then $\rho\in\textrm{SEP}_{N_{1}\otimes N_{2}}$ if $\lambda\leq\bar{\lambda}(E_{\rho})\equiv\frac{1}{\sum_{k=1}^{K}\frac{e_{k}}{\lambda^{*}(e_{k})}}$, where $\lambda^{*}(e_{k})=\frac{1}{1+r_{1}r_{2}\sigma_{0}(e_{k})\sigma_{1}(e_{k})}$ and $\sigma_{i}(e_{k})$ are the generalized Schmidt coefficients of $\ket{e_{k}}$ with respect to $\left((\bar{F}_{1}\bar{F}_{1}^{\dagger})^{-1},(\bar{F}_{2}\bar{F}_{2}^{\dagger})^{-1}\right)$. If $K=1$ the condition is also sufficient.
\end{result}

\begin{proof}
The proof runs along similar lines to that of previous sections. We can decompose

\begin{equation}
\rho=(1-\lambda)M_{1}\otimes M_{2}+\lambda\ket{z}\bra{z}=\bar{F}_{1}\otimes\bar{F}_{2}\left[(1-\lambda)\frac{\left(\mathbb{I}_{r_1}\oplus\mathbb{O}_{N-r_1}\right)\otimes\left(\mathbb{I}_{r_2}\oplus\mathbb{O}_{N-r_2}\right)}{r_{1}r_{2}}+\lambda\ket{\sigma}\bra{\sigma}\right]\bar{F}_{1}^{\dagger}\otimes\bar{F}_{2}^{\dagger},
\end{equation}

\noindent where $\ket{\sigma}=\sum_{j=0}^{r-1}\sigma_{j}(z)\ket{jj}$, $\sigma_{j}(z)$ being the generalized Schmidt coefficients of $\ket{z}$ with respect to $\left((\bar{F}_{1}\bar{F}_{1}^{\dagger})^{-1},(\bar{F}_{2}\bar{F}_{2}^{\dagger})^{-1}\right)$. Using the same decomposition as in appendix \ref{Lem}, then $(1-\lambda)\frac{\left(\mathbb{I}_{r_1}\oplus\mathbb{O}_{N-r_1}\right)\otimes\left(\mathbb{I}_{r_2}\oplus\mathbb{O}_{N-r_2}\right)}{r_{1}r_{2}}+\lambda\ket{\sigma}\bra{\sigma}$ is separable if, and only if, $\lambda\leq\lambda^{*}_{F_{\mathcal{S}_{\rho}}}(z)=\frac{1}{1+r_{1}r_{2}\sigma_{0}(z)\sigma_{1}(z)}$.\\

In those cases in which $E_{\rho}=\sum_{k=1}^{K}e_{k}\ket{e_{k}}\bra{e_{k}}$ the proof is exactly the same as that of result \ref{IncomRes}. In these cases, the condition is only necessary.
\end{proof}

This result only applies to those quantum states belonging to faces which contain at least a product matrix of the form $M_{1}\otimes M_{2}$, which are those faces $F_{\mathcal{S}}$ parametrized by tensor product linear subspaces, i.e. faces of the form $F_{\mathcal{S}}=F_{\mathcal{S}_{1}\otimes\mathcal{S}_{2}}$. For this reason result \ref{ProFaces} is the second incomplete result of this approach.

\section{Plausible generalizations to multipartite systems}
\label{PlaGenMulSys}

Since the preceding results are independent of the dimensions of the quantum subsystems, they can be extended to the multipartite realm to detect genuine entanglement in the so-called generalized Werner states:

\begin{result}
Let $\rho=(1-\lambda)M_{1}\otimes\dots\otimes M_{n}+\lambda\ket{z}\bra{z}$, where $M_{i}=F_{i}F_{i}^{\dagger}$ and $r_{M_{i}}=r_{i}$, $i=1,\dots,n$. Then $\rho$ is genuinely entangled if, and only if,

\begin{equation}
\lambda>\lambda^{*}(z)=\frac{1}{1+\left(\prod_{i=1}^{n}r_{i}\right)\min_{k=1,\dots, n-1}\{\sigma_{0}(z^{i_{1}\dots i_{k}|i_{k+1}\dots i_{n}})\sigma_{1}(z^{i_{1}\dots i_{k}|i_{k+1}\dots i_{n}})},
\end{equation}

\noindent where $\sigma_{j}(z^{i_{1}\dots i_{k}|i_{k+1}\dots i_{n}})$ are the generalized Schmidt coefficients of $\ket{z}$ with respect to\\ $\left(\bigotimes_{j=i_{1}}^{i_{k}}(\bar{F}_{j}\bar{F}_{j}^{\dagger})^{-1},\bigotimes_{j=i_{k+1}}^{i_{n}}(\bar{F}_{j}\bar{F}_{j}^{\dagger})^{-1}\right)$ viewed as a bipartite quantum state of systems $i_{1}\dots i_{k}$ and the rest $i_{k+1}\dots i_{n}$.

\end{result}

\begin{proof}
The result readily follows after applying result \ref{ResSufNec} to each bipartition of state $\ket{z}$, since to be genuinely entangled it cannot be biseparable for any bipartition \cite{EisGro05a}.
\end{proof}

This generalizes previous results in the literature \cite{PitRub00a,DenChe08a}. However, notice that we can only conclude about the genuine entanglement of multipartite quantum states and not about their separability, since in this case it is necessary to distinguish between full separability and $m-$order separability \cite{EisGro05a}. This will be undertaken elsewhere.\\

As an illustration of this result let us focus upon three-partite systems and, in particular, upon states of the form $\rho=(1-\lambda)\frac{\mathbb{I}_{8}}{8}+\lambda\ket{\textrm{GHZ}}\bra{\textrm{GHZ}}$, where $\ket{\textrm{GHZ}}$ is the well-known GHZ state \cite{GreHorShiZei90a}. Then it is immediate to conclude that $\rho$ is genuinely entangled if, and only if, $\lambda>\frac{1}{1+8\cdot\frac{1}{2}}=\frac{1}{5}$. Analogously, if $\rho=(1-\lambda)\frac{\mathbb{I}_{8}}{8}+\lambda\ket{\textrm{W}}\bra{\textrm{W}}$, where $\ket{\textrm{W}}$ is the W state \cite{DurVidCir00a}, $\rho$ is genuinely entangled if, and only if, $\lambda>\frac{1}{1+8\cdot\sqrt{\frac{2}{3}}\frac{1}{\sqrt{3}}}=\frac{1}{1+\frac{8\sqrt{2}}{3}}$.\\

Furthermore we see under this approach a natural connection of the SLOCC classification of pure quantum states with the question of separability of mixed states. Let $\ket{\Psi}$ and $\ket{z}$ be connected through the invertible local operations (ILOs) $G_{i}$, $i=1,\dots,n$, i.e. $\ket{z}=G_{1}\otimes\dots\otimes G_{n}\ket{\Psi}$ (see \cite{DurVidCir00a}). Then the multipartite state $\rho=(1-\lambda)M_{1}\otimes\dots\otimes M_{n}+\lambda\ket{z}\bra{z}$, where $M_{i}=F_{i}F_{i}^{\dagger}$ and $r_{M_{i}}=r_{i}$, is genuinely entangled if, and only if, the state $$\bar{\rho}=(1-\lambda)\frac{\left(\mathbb{I}_{r_{1}}\oplus\mathbb{O}_{N_{1}-r_{1}}\right)\otimes\dots\otimes\left(\mathbb{I}_{r_n}\oplus\mathbb{O}_{N_{n}-r_{n}}\right)}{r_{1}\dots r_{n}}+\lambda \bar{F}_{1}^{-1}G_{1}\otimes\dots\otimes \bar{F}_{n}^{-1}G_{n}\ket{\Psi}\bra{\Psi}G^{\dagger}_{1}\bar{F}^{-\dagger}_{1}\otimes\dots\otimes G^{\dagger}_{n}\bar{F}^{-\dagger}_{n}$$\noindent is genuinely entangled, which will be the case whenever $$\lambda>\frac{1}{1+r_{1}\dots r_{n}\min_{k=1,\dots,n-1}\{\sigma_{0}^{i_{1}\dots i_{k}|i_{k+1}\dots i_{n}}\sigma_{1}^{i_{1}\dots i_{k}|i_{k+1}\dots i_{n}}\}},$$\noindent where with relaxed notation $\sigma_{j}^{i_{1}\dots i_{k}|i_{k+1}\dots i_{n}}$ denote the generalized Schmidt coefficients of $\bar{F}_{1}^{-1}G_{1}\otimes\dots\otimes \bar{F}_{n}^{-1}G_{n}\ket{\Psi}$. In this case the canonical states $\ket{\Psi}$ of each class should be chosen with the additional criterion of easing the computation of these generalized Schmidt coeficcients. In this sense, the states with permutation symmetry are excellent candidates \cite{BasKriMatGodLamSol09a}. To illustrate this point consider a n-qubit $W$ state $\ket{W_{n}}=\frac{1}{\sqrt{n}}\left[\ket{10\dots 0}+\ket{01\dots0}+\dots+\ket{00\dots1}\right]$. The state $(1-\lambda)\frac{\mathbb{I}_{2^{n}}}{2^{n}}+\lambda\ket{W_{n}}\bra{W_{n}}$ will be genuinely entangled if, and only if, $\lambda>\frac{1}{1+2^{n}\frac{\sqrt{n-1}}{n}}$. Another example arises when considering the generalized $n$-qudit GHZ state $\ket{GHZ_{n}^{(d)}}=\frac{1}{\sqrt{d}}\sum_{i=0}^{d-1}\ket{i\dots i}$. The state $(1-\lambda)\frac{\mathbb{I}_{d^{n}}}{d^{n}}+\lambda\ket{GHZ_{n}^{(d)}}\bra{GHZ_{n}^{(d)}}$ will be genuinely entangled if, and only if, $\lambda>\frac{1}{1+d^{n-1}}$. These results allow us to investigate the $d$-dimensional multipartite genuine entanglement as the number of qudits increases: we must focus on the asymptotic behaviour of $g(n)\equiv d^{n}\min_{k=1\dots,n-1}\{\sigma_{0}^{i_{1}\dots i_{k}|i_{k+1}\dots i_{n}}\sigma_{1}^{i_{1}\dots i_{k}|i_{k+1}\dots i_{n}}\}$. For instance, if $g(n)=O(1)$, then the amount of entanglement (in the sense of the amount of permitted noise in the pure  state preparation) does not change with $n$. We conjecture that indeed the quantity $g(n)$ can be the basis for a multiqudit entanglement measure.

\section{Conclusions}
\label{Con}

We can summarize the preceding results in the following partial separability check:

\begin{summary}
Let $\rho$ be a mixed bipartite quantum state of dimensions $N_{1}$ and $N_{2}$:
\begin{enumerate}
\item[i.] It can always be decomposed as
    \begin{equation}\label{DetDec}
    \rho=(1-\lambda)C_{\rho}+\lambda\left[\sum_{k=1}^{K}e_{k}\ket{e_{k}}\bra{e_{k}} \right],
    \end{equation}
    \noindent where $C_{\rho}$ is an arbitrary matrix such that $r_{C_{\rho}}=r_{\rho}$, $\{\ket{e_{k}}\}$ are linear independent vectors and $K<r_{\rho}$.
\item[ii.] Let $\rho=(1-\lambda)M_{1}\otimes M_{2}+\lambda\left[\sum_{k=1}^{K}e_{k}\ket{e_{k}}\bra{e_{k}} \right]$, where $M_{i}=F_{i}F_{i}^{\dagger}$, $i=1,2$ and $r_{i}\equiv r_{M_{i}}$. Let $\bar{F}_{i}$ be the orthogonally complemented matrix of $F_{i}$, $i=1,2$. Then if
    \begin{equation}\label{ConSep}
    \lambda\leq\frac{1}{\sum_{k=1}^{K}e_{k}\left[1+r_{1}r_{2}\sigma_{0}(e_{k})\sigma_{1}(e_{k})\right]},
    \end{equation}
    \noindent where $\sigma_{j}(e_{k})$ are the generalized Schmidt coefficients of $\ket{e_{k}}$ with respect to\\ $\left((\bar{F}_{1}\bar{F}_{1}^{\dagger})^{-1},(\bar{F}_{2}\bar{F}_{2}^{\dagger})^{-1}\right)$, the quantum state $\rho$ is separable.
\item[iii.] In the same conditions, if $K=1$, the condition is also sufficient.
\end{enumerate}
\end{summary}

Reminding the cone structure of the set of unnormalized mixed quantum states, the previous results show two limitations. On the one hand, results (ii) and (iii) only apply to those states beloging to faces $F_{\mathcal{S}}$ with product structure $F_{\mathcal{S}}=F_{\mathcal{S}_{1}\otimes\mathcal{S}_{2}}$. On the other hand, when the universal decomposition \eqref{DetDec} is such that $K>1$, the condition \eqref{ConSep} is only necessary.\\

Despite these limitations, the approach is not exhausted, since many other choices for the matrix $C_{\rho}$ exist. We believe that the present limitations arise from the simple choice of $C_{\rho}$ as a tensor product matrix and that a more complex structure will overcome both of them. This is currently under investigation.\\

The independence of this approach with respect to the dimensions allows us to give the first steps to extend the results to multipartite systems. And in particular, result (iii) can be extended to detect genuine entanglement in generalized Werner states:

\begin{summary}
Let $\rho=(1-\lambda)M_{1}\otimes\dots\otimes M_{n}+\lambda\ket{z}\bra{z}$, where $M_{i}=F_{i}F_{i}^{\dagger}$, $i=1,\dots,n$ and $r_{i}\equiv r_{M_{i}}$. Let $\bar{F}_{i}$ be the orthogonally complemented matrix of $F_{i}$, $i=1,\dots,n$. Then $\rho$ is genuinely entangled if, and only if,
\begin{equation}\label{MulCon}
\lambda>\frac{1}{1+\left(\prod_{i=1}^{n}r_{i}\right)\min_{k=1,\dots, n-1}\{\sigma_{0}(z^{i_{1}\dots i_{k}|i_{k+1}\dots i_{n}})\sigma_{1}(z^{i_{1}\dots i_{k}|i_{k+1}\dots i_{n}})},
\end{equation}

\noindent where $\sigma_{j}(z^{i_{1}\dots i_{k}|i_{k+1}\dots i_{n}})$ are the generalized Schmidt coefficients of $\ket{z}$ with respect to $$\left(\bigotimes_{j=i_{1}}^{i_{k}}(\bar{F}_{j}\bar{F}_{j}^{\dagger})^{-1},\bigotimes_{j=i_{k+1}}^{i_{n}}(\bar{F}_{j}\bar{F}_{j}^{\dagger})^{-1}\right)$$ viewed as a bipartite quantum state of systems $i_{1}\dots i_{k}$ and the rest $i_{k+1}\dots i_{n}$.
\end{summary}

Condition \eqref{MulCon} only detects genuine entanglement and does not resolve between full separability and $m-$order separability. The main reason lies on the fact that decomposition \eqref{DetDec} is thought to detect bipartite entanglement. Thus to distinguish between the different sorts of multipartite entanglement, we must adapt the choices of both $C_{\rho}$ and $E_{\rho}$ to this situation. This will be undergone elsewhere. However, notice that the approach is not exhausted once more.\\

Finally a natural connection between the multipartite separability problem and the pure state classification under SLOCC arises, since the entanglement of any state admitting the choice $C_{\rho}=M_{1}\otimes\dots\otimes M_{n}$ is equivalent to that of $\bar{\rho}=(1-\lambda)\frac{\mathbb{I}_{N_{1}\dots N_{n}}}{N_{1}\dots N_{n}}+\lambda H_{1}\otimes\dots\otimes H_{n}\ket{\Psi}\bra{\Psi}H_{1}^{\dagger}\otimes\dots\otimes H_{n}^{\dagger}$, where $\ket{\Psi}$ is the canonical state of the particular SLOCC class. An appropiate choice of $\ket{\Psi}$ will drive us to an easier computation of the generalized Schmidt coefficients of $H_{1}\otimes\dots\otimes H_{n}\ket{\Psi}$.\\


\section*{Acknowledgments}

The authors acknowledge financial support from the Spanish MICINN project No.\ MICINN-08-FIS2008-00288.

\appendix

\section{The generalized singular value decomposition}
\label{GSVD}

The singular value decomposition (SVD) is a standard tool in matrix analysis (cf. e.g. \cite{HorJoh91a}) by which any rectangular matrix $M\in\mathcal{M}_{m,n}(\mathbb{C})$ can be written as

\begin{equation}\label{SVD}
M=UDV^{\dagger}=\sum_{k=1}^{r}\sigma_{k}\ket{u_{k}}\bra{v_{k}},
\end{equation}

\noindent where $\sigma_{1}\geq\sigma_{2}\geq\dots\geq\sigma_{r}>0$ are the singular values, the matrix $U$, whose columns $\ket{u_{i}}$, $i=1,\dots,m$, are the so-called left singular vectors, is unitary and the matrix $V$, whose columns $\ket{v_{j}}$, $j=1,\dots,n$, are the so-called right singular vectors, is also unitary. The integer $r$ is indeed the rank of the matrix $M$.\\

The generalized singular value decomposition (GSVD) is a possible extension of this result (see e.g. \cite{Loa76a}) which states that, given two $m-$ and $n-$dimensional positive definite matrices $A$ and $B$, any rectangular matrix $M\in\mathcal{M}_{m,n}(\mathbb{C})$ can be written as

\begin{equation}\label{SVD}
M=\bar{U}D\bar{V}^{\dagger}=\sum_{k=1}^{r}\sigma_{k}^{A,B}\ket{u_{k}^{A,B}}\bra{v_{k}^{A,B}},
\end{equation}

\noindent where $\sigma_{1}^{A,B}\geq\sigma_{2}^{A,B}\geq\dots\geq\sigma_{r}^{A,B}>0$ are the generalized singular values with respect to $(A,B)$, the matrix $\bar{U}$, whose columns $\ket{u_{i}^{A,B}}$, $i=1,\dots,m$, are the so-called left generalized singular vectors, is $A$-unitary \footnote{$\bar{U}^{\dagger}A\bar{U}=\mathbb{I}_{m}$.} and the matrix $V$, whose columns $\ket{v_{j}^{A,B}}$, $j=1,\dots,n$, are the so-called right singular vectors, is $B$-unitary \footnote{$\bar{V}^{\dagger}B\bar{V}=\mathbb{I}_{n}$.}. The integer $r$ is again the rank of the matrix $M$. In the text we have dropped out the specification $(A,B)$ to ease the notation.\\

Notice that the difference between the SVD and the GSVD is just the scalar product of the rows and the columns of $M$ under which the standard proof of the decomposition is established (see \cite{HorJoh91a}).

\section{Lemmas}
\label{Lem}

We gather some technical lemmas necessary to prove the preceding results.

\begin{lemma}\label{CompRoots}
Let $\omega^{k}$, where $k=0,1,\dots,n-1$, be the $n$th complex roots of the unit. Then
\begin{equation}
\sum_{k=0}^{n-1}[\omega^{k}]^{p}=\left\{\begin{array}{ll}
0& \forall p\quad\textrm{ such that }\frac{p}{n}\notin\mathbb{N},\\
n &\forall p\quad\textrm{ such that }\frac{p}{n}\in\mathbb{N}.
\end{array}\right.
\end{equation}
\end{lemma}
\begin{proof}
Elementary \cite{ChuBro08a}.
\end{proof}

\begin{lemma}\label{Roots}
Let $I=\{0,\dots,r-1\}$ and $\omega=e^{\frac{2\pi i}{n}}$. Let  $\omega_{0}=1$, $\omega_{1}=\omega^{*}$, $\omega_{2}=\omega^{2}$, $\omega_{3}=\omega^{*4}$, \dots,$\omega_{m}=\omega^{(*)2^{m}}$,\dots, $m\geq 1$, be complex phases where the complex conjugation appears only in the odd indices. Then for all $i,j,p,q\in I$ we have
\begin{equation}
\omega_{i}\omega_{j}^{*}\omega_{p}^{*}\omega_{q}=\left\{\begin{array}{ll}
1& \textrm{ if } i=j, p=q\qquad\textrm{or}\qquad i=p,j=q\quad,\\
\omega^{m_{ijpq}} & \textrm{ otherwise, where } m_{ijpq}\in\mathbb{Z}-\{0\}.
\end{array}\right.
\end{equation}
\end{lemma}

\begin{proof}
 The expression $\omega_{i}\omega_{j}^{*}\omega_{p}^{*}\omega_{q}$, with $\omega_{l}=1,\omega^{*},\omega^{2},\omega^{4*},\omega^{8},\omega^{16*},\dots$, reduces to $1$ only in the stated cases, that is, if $i=j, p=q$ or $i=p,j=q$. In the rest of cases, it is (almost) evident that the expression reduces to a power of $\omega$ or of $\omega^{*}$.
\end{proof}

\begin{lemma}\label{DesigNumer}
Let $\sigma_{0}\geq\sigma_{1}\geq\dots\geq\sigma_{r-1}>0$ be
 such that $\sum_{i=0}^{r-1}\sigma_{i}^{2}=1$. Then we have

\begin{equation}
\sigma_{0}\sigma_{1}\sum_{k=0}^{r-1}\ket{k}\bra{k}+\sigma_{i}^{2}\ket{i}\bra{i}\geq\sigma_{i}^{2}\ket{i}\bra{i}+\sum_{\substack{k=0\\k\neq i}}^{r-1}\sigma_{i}\sigma_{k}\ket{k}\bra{k},\qquad \forall i=0,\dots,r-1.
\end{equation}
\end{lemma}
\begin{proof}
This is equivalent to prove the following two numeric inequalities for all  $i=0,\dots,r-1$:
\begin{eqnarray}
\sigma_{0}\sigma_{1}+\sigma_{i}^{2}&\geq&\sigma_{i}^{2},\nn\\
\sigma_{0}\sigma_{1}&\geq &\sigma_{i}\sigma_{k},\qquad \forall k\neq i.\nn
\end{eqnarray}

Both are elementary.

\end{proof}

\begin{lemma}\label{PPT}
Let $Q=(1-\lambda)\frac{\mathbb{I}_{N_{1}N_{2}}}{N_{1}N_{2}}+\lambda\ket{\sigma}\bra{\sigma}$, where $\ket{\sigma}=\sum_{i=0}^{r-1}\sigma_{i}\ket{ii}$. Then $Q^{T_{1}}\geq 0$ if, and only if, $\lambda\leq\frac{1}{1+N_{1}N_{2}\sigma_{0}\sigma_{1}}$.
\end{lemma}

\begin{proof}
The partial transpose of $Q$ is given by $$Q^{T_{1}}=(1-\lambda)\frac{\mathbb{I}_{N_{1}N_{2}}}{N_{1}N_{2}}+\lambda\left[\sum_{i=0}^{r-1}\sigma_{i}^{2}\ket{ii}\bra{ii}+\sum_{i=0}^{r-1}\sum_{\substack{j=0\\j\neq i}}^{r-1}\sigma_{i}\sigma_{j}\ket{s_{ij}}\bra{s_{ij}}-\sum_{i=0}^{r-1}\sum_{\substack{j=0\\j\neq i}}^{r-1}\sigma_{i}\sigma_{j}\ket{a_{ij}}\bra{a_{ij}}\right],$$

\noindent where $\ket{s_{ij}}=\frac{1}{\sqrt{2}}\left[\ket{ij}+\ket{ji}\right]$ and $a_{ij}=\frac{1}{\sqrt{2}}\left[\ket{ij}-\ket{ji}\right]$. Now the system of vectors $\{\ket{ii}\}_{i=0}^{\min(N_{1}-1,N_{2}-1)}\cup\{\ket{s_{ij}},\ket{a_{ij}}\}_{i,j=0}^{\min(N_{1}-1,N_{2}-1)}\cup\{\ket{s_{ij}},\ket{a_{ij}}\}_{i \textrm{ or } j=0,\dots,\min(N_{1}-1,N_{2}-1)}^{j\textrm{ or } i=\min(N_{1}-1,N_{2}-1)+1,\dots,\max(N_{1}-1,N_{2}-1)}$ is an orthonormal basis in $\mathbb{C}^{N_{1}}\otimes\mathbb{C}^{N_{2}}$. Thus the identity matrix $\mathbb{I}_{N_{1}N_{2}}$ admits a spectral decomposition in terms of the unidimensional orthogonal projectors onto these vectors. Summing up all the terms we have the spectral decomposition of $Q^{T_{1}}$ with eigenvalues of the alternative forms $\frac{1-\lambda}{N_{1}N_{2}}+\lambda\sigma_{i}^{2}$, $\frac{1-\lambda}{N_{1}N_{2}}+\lambda\sigma_{i}\sigma_{j}$ and $\frac{1-\lambda}{N_{1}N_{2}}-\lambda\sigma_{i}\sigma_{j}$, all of them will be nonnegative provided $\frac{1-\lambda}{N_{1}N_{2}}-\lambda\sigma_{i}\sigma_{j}\geq 0$ for all $i\neq j$. Since the generalized Schmidt coefficients constitutes a nonincreasing sequence, then $\lambda\geq\frac{1}{1+N_{1}N_{2}\sigma_{0}\sigma_{1}}$.

\end{proof}

For the next result we need to define $M=\max_{ijpq\in I}\{m_{ijpq}\}$ and $n_{0}=M+1$, where $m_{ijpq}$ and $I$ are defined as in lemma \ref{Roots}, $r\leq\min(N_{1}-1,N_{2}-1)$ and $\sigma_{i}$ are normalized and in non-increasing order. Define $\omega=e^{\frac{2\pi i}{n_{0}}}$.

\begin{lemma}\label{MatIden}
We have the identity

\begin{eqnarray}
\sum_{i=0}^{r-1}\ket{i}\bra{i}\otimes\left(\sigma^{2}_{i}\ket{i}\bra{i}+\sum_{\substack{k=0\\k\neq i}}^{r-1}\sigma_{i}\sigma_{k}\ket{k}\bra{k}\right)+\sum_{i=0}^{r-1}\sum_{\substack{j=0\\j\neq i}}^{r-1}\ket{i}\bra{j}\otimes\sigma_{i}\sigma_{j}\ket{i}\bra{j}&=&\nn\\
\frac{\left(\sum_{i=0}^{r-1}\sigma_{i}\right)^{2}}{n_{0}}\sum_{k=0}^{n_{0}-1}\ket{u_{k}}\bra{u_{k}}\otimes\ket{u_{k}^{*}}\bra{u_{k}^{*}},&&
\end{eqnarray}
\noindent where $\ket{u_{k}}=\sum_{i=0}^{r-1}\omega_{i}^{k}\sigma_{i}^{1/2}\ket{i}$ and $\ket{u_{k}^{*}}=\sum_{i=0}^{r-1}\omega_{i}^{*k}\sigma_{i}^{1/2}\ket{i}$.

\end{lemma}

\begin{proof}
Taking into account the definitions we have

\begin{gather}
\left(\sum_{i=0}^{r-1}\sigma_{i}\right)^{2}\ket{u_k}\bra{u_{k}}\otimes\ket{u_{k}^{*}}\bra{u_{k}^{*}}=\sum_{ijpq=0}^{r-1}\omega_{i}^{k}\omega_{j}^{*k}\omega_{p}^{*k}\omega_{q}^{k}\sqrt{\sigma_{i}\sigma_{j}\sigma_{p}\sigma_{q}}\ket{i}\bra{j}\otimes \ket{p}\bra{q},
\end{gather}

\noindent which after summing up in the powers of $k$ yields
\begin{gather}
\sum_{ijpq=0}^{r-1}\left(\sum_{k=0}^{n_{0}-1}\omega_{i}^{k}\omega_{j}^{*k}\omega_{p}^{*k}\omega_{q}^{k}\right)\sqrt{\sigma_{i}\sigma_{j}\sigma_{p}\sigma_{q}}\ket{i}\bra{j}\otimes \ket{p}\bra{q},
\end{gather}

Taking into account lemmas  \ref{CompRoots} and \ref{Roots} for the choice of the base  $n_{0}$ of the roots, we will have \footnote{With $n_{0}=M+1$ for the complex roots of the unit, it is clear that $p/n_{0}$ will never be a natural number for any value $p\leq M$.}

\begin{gather}
\sum_{ij=0}^{r-1}n_{0}\sigma_{i}\sigma_{j}\ket{i}\bra{i}\otimes \ket{j}\bra{j}+\sum_{\substack{ij=0\\i \neq j}}^{r-1}n_{0}\sigma_{i}\sigma_{j}\ket{i}\bra{j}\otimes \ket{i}\bra{j},
\end{gather}

Dividing by $n_{0}$ we finally get the desired result.

\end{proof}

\begin{lemma}\label{SepWerner}
With the same definitions as above, we have
\begin{equation}
\sigma_{0}\sigma_{1}\mathbb{I}_{N_{1}N_{2}}+\sum_{i=1}^{r}\sum_{j=1}^{r}\sigma_{i}\sigma_{j}\ket{i}\bra{j}\otimes \ket{i}\bra{j}\in SEP_{N_{1}\otimes N_{2}}.
\end{equation}
\end{lemma}
\begin{proof}

Firstly, since $\mathbb{I}_{N_{1}N_{2}}=\sum_{i=0}^{N_{1}-1}\ket{i}\bra{i}\otimes\sum_{j=0}^{N_{2}-1}\ket{j}\bra{j}$ we can write

\begin{gather}
\sigma_{0}\sigma_{1}\mathbb{I}_{N_{1}N_{2}}+\sum_{i=0}^{r-1}\sum_{j=0}^{r-1}\sigma_{i}\sigma_{j}\ket{i}\bra{j}\otimes \ket{i}\bra{j}=\nn\\
=\sigma_{0}\sigma_{1}\sum_{i=0}^{r-1}\ket{i}\bra{i}\otimes\sum_{j=r}^{N_{2}-1}\ket{j}\bra{j}+\sigma_{0}\sigma_{1}\sum_{i=r}^{N_{1}-1}\ket{i}\bra{i}\otimes\sum_{j=0}^{r-1}\ket{j}\bra{j}+\sigma_{0}\sigma_{1}\sum_{i=r}^{N_{1}-1}\ket{i}\bra{i}\otimes\sum_{j=r}^{N_{2}-1}\ket{j}\bra{j}+\nn\\
+\left(\sigma_{0}\sigma_{1}\sum_{ij=0}^{r-1}\ket{i}\bra{i}\otimes \ket{j}\bra{j}+\sum_{i=0}^{r-1}\sum_{j=0}^{r-1}\sigma_{i}\sigma_{j}\ket{i}\bra{j}\otimes \ket{i}\bra{j}\right).
\end{gather}

Since the three first matrices in the right hand side are clearly separable, we will focus on the last one, which we rewrite as

\begin{gather}
\sum_{i=0}^{r-1}\ket{i}\bra{i}\otimes\left(\sigma_{0}\sigma_{1}\sum_{j=0}^{r-1}\ket{j}\bra{j}+\sigma_{i}^{2}\ket{i}\bra{i}\right)+\sum_{i=0}^{r-1}\sum_{\substack{j=0\\j\neq i}}^{r-1}\sigma_{i}\sigma_{j}\ket{i}\bra{j}\otimes \ket{i}\bra{j}.
\end{gather}

In view of lemma \ref{DesigNumer}, we can write

\begin{gather}
\sum_{i=0}^{r-1}\ket{i}\bra{i}\otimes\left(\sigma_{0}\sigma_{1}\sum_{j=0}^{r-1}\ket{j}\bra{j}+\sigma_{i}^{2}\ket{i}\bra{i}\right)+\sum_{i=0}^{r-1}\sum_{\substack{j=0\\j\neq i}}^{r-1}\sigma_{i}\sigma_{j}\ket{i}\bra{j}\otimes \ket{i}\bra{j}\geq \nn\\
\sum_{i=0}^{r-1}\ket{i}\bra{i}\otimes\left(\sigma_{i}^{2}\ket{i}\bra{i}+\sum_{\substack{j=0\\j\neq i}}^{r-1}\sigma_{j}\sigma_{i}\ket{j}\bra{j}\right)+\sum_{i=0}^{r-1}\sum_{\substack{j=0\\j\neq i}}^{r-1}\sigma_{i}\sigma_{j}\ket{i}\bra{j}\otimes \ket{i}\bra{j},
\end{gather}

\noindent where the difference is clearly a separable matrix. Now, this is exactly the matrix in lemma \ref{MatIden}, which is separable.

\end{proof}


\end{document}